\documentclass[final,12pt,3p,onecolumn]{elsarticle}
\makeatletter
\def\ps@pprintTitle{%
	\let\@oddhead\@empty
	\let\@evenhead\@empty
	\def\@oddfoot{\centerline{\thepage}}%
	\let\@evenfoot\@oddfoot}
\makeatother
\usepackage[colorlinks = true,
linkcolor = blue,
urlcolor  = blue,
citecolor = blue,
anchorcolor = blue]{hyperref}
\usepackage{amsmath}
\usepackage{amssymb}
\usepackage{xcolor}
\newtheorem{theorem}{Theorem}
\newtheorem{lemma}{Lemma}
\newtheorem{remark}{Remark}

\newtheorem{proof}{Proof}

\usepackage{dcolumn}
\usepackage{subfigure}
\usepackage{graphicx}

\usepackage{epstopdf}
\usepackage{epsfig}

\begin{document}
	
	\begin{frontmatter}
		
		\title{Reduction of damped, driven Klein-Gordon equations into a discrete nonlinear Schr\"odinger equation: justification and numerical comparisons}

		\author[label1,label5]{Y.\ Muda}
		\address[label1]{Department of Mathematical Sciences,University of Essex, Colchester, CO4 3SQ, United Kingdom}
		\address[label5]{Department of Mathematics, Faculty of Science and Technology,\\ 
			State Islamic University of Sultan Syarif Kasim Riau, Pekanbaru, 28294,  Indonesia\fnref{label4}}
		\address[label2]{Theoretical Physics Laboratory,Theoretical High Energy Physics and Instrumentation Research Group, Faculty of Mathematics and Natural Sciences,  Institut Teknologi Bandung,  Bandung, Indonesia, 40132\fnref{label4}}
		\address[label6]{Industrial and Financial Mathematics Research Group, Faculty of Mathematics and Natural Sciences, Institut Teknologi Bandung, Bandung, 40132, Indonesia.}
		
		\cortext[cor1]{Corresponding author}
		
		\ead{ymuda@essex.ac.uk}
		
		\author[label2]{F.T.\ Akbar}
		\ead{ftakbar@fi.itb.ac.id}
		
		\author[label1,label6]{R.\ Kusdiantara}
		\ead{rudy@math.itb.ac.id}
		
		\author[label2]{B.E.\ Gunara}
		\ead{bobby@fi.itb.ac.id}
		
		\author[label1]{H.\ Susanto\corref{cor1}}
		\ead{hsusanto@essex.ac.uk}
		
		\begin{abstract}
			We consider a discrete nonlinear Klein-Gordon equations with damping and external drive. Using a small amplitude ansatz, one usually approximates the equation using a damped, driven discrete nonlinear Schr\"odinger equation. Here, we show for the first time the justification of this approximation by finding the error bound using energy estimate. Additionally, we prove the local and global existence of the Schr\"odinger equation. Numerical simulations are performed that describe the analytical results. Comparisons between discrete breathers of the Klein-Gordon equation and discrete solitons of the discrete nonlinear Schr\"odinger equation are presented.
		\end{abstract}
		
		
	\end{frontmatter}
	
	
	\section{Introduction}
	\label{sec1}
	
	Nonlinear lattices are a set of nonlinear evolution equations that are coupled spatially. Prominent classes of nonlinear lattices are discrete Klein-Gordon \cite{1} and Frenkel-Kontorova  equations \cite{6} that serve as possibly the simplest models for many complex physical and biological systems. While Frenkel-Kontorova type systems correspond to coupled equations with harmonic on-site potential, the discrete Klein-Gordon equations do with the anharmonic one. 
	
	Small-amplitude wave packets of nonlinear lattices are usually explored via reduction to an amplitude or modulation equation in the form of either continuous or discrete nonlinear Schr\"odinger equations. The method is usually referred to as the rotating wave approximation. If one is interested in solution profiles with a much larger length scale than the typical distance between the lattices, they can aim for a continuous approximation and obtain nonlinear Schr\"odinger equations (see, e.g., \cite{14,13,4}). When one is instead interested in waves with the same scale of the typical lattice distance, one will obtain an approximation in the form of a discrete nonlinear Schr\"odinger equation (see, e.g., \cite{19,11,7,15,16,12,23}) with the corresponding wave properties quite distinct from those in the continuous limit. 
	
	Despite their widespread use, rigorous justifications of the rotating wave procedures are more sparse, with an early example being \cite{14}, wherein Hamiltonian Klein-Gordon lattices are approximated by nonlinear Schr\"odinger equations (see also \cite{9,17}). A justification for the discrete nonlinear Schr\"odinger approximation was provided rather recently in \cite{22}.
	
	In this paper, we consider a Klein-Gordon equation with external damping and drive. Our present work will be relevant to models appearing in the study of, e.g., superconducting Josephson junctions \cite{3}, mechanical oscillators \cite{18}, electrical lattices \cite{10}, etc. Using the rotating-wave approximation, we will show that the corresponding modulation equation is a damped, driven discrete nonlinear Schr\"odinger equation. The similar approximation has been used to reduce an externally driven sine-Gordon equation into a damped, driven continuous nonlinear Schr\"odinger equation \cite{24,5}. In here, we are going to provide a rigorous justification of the discrete Schr\"odinger equation. Note that our work here is significantly different from the aforementioned published works in the sense that our original governing equation as well as the modulation one are not Hamiltonian. Moreover, the external drive yields a constant term that can be challenging to control in providing boundedness of the error, i.e., the solution does not lie in $\ell^2$-space of $\mathbb{Z}$. Without external damping and drive, the initial value problem for the discrete nonlinear Schr\"odinger equation with power nonlinearity in weighted $\ell^2$-space has been shown to be globally well-posed in \cite{21}. N'Gu\'er\'ekata and Pankov \cite{20} provides a stronger result of global well-posedness in spaces of exponentially decaying data. Here, by considering the problem with damping and drive in a periodic domain, we are able to provide the global existence of the discrete nolinear Schr\"dinger equation as well as the error bound of the rotating wave approximation. 
	
	This paper is organized as follows. We define the governing equation and formulate preliminary results on the unique global solution and error estimate in Section \ref{sec2}. The main result on the error bound of the rotating-wave approximation as time evolves is presented in Section \ref{sec3}. In Section \ref{sec4} we describe the computation of the error made by the Schr\"odinger approximation numerically. The comparison is provided for localised waves, i.e., breather solutions. 
	
	\section{Mathematical formulation and preliminary results}
	\label{sec2}
	
	Consider the following model of coupled oscillators with damping and drive on a finite lattice 
	\begin{equation}\label{dKGdriven}
	\ddot{u}_j = -u_j - \xi u_j^3 + \epsilon(\Delta_2 u_j) - \alpha \dot{u}_j + \frac{h}{2}(e^{i\Omega t} + e^{-i\Omega t}),\quad j\in\mathbb{Z}_{N} = \{1,\dots,N\},
	\end{equation}
	{where $u_j \equiv u_j(t)$ is a real-valued wave function at site $j$, the overdot is the time derivative and $\epsilon$ represents the coupling constant between two adjacent sites, with $\Delta_2 u_j = u_{j+1} - 2u_j + u_{j-1}$ being the discrete Laplacian in one dimension. The parameters $\alpha$ and $h$ denote the damping coefficient and the strength of the external drive, respectively. The driving frequency is taken to be $\Omega = 1-\frac{\epsilon \omega}{2}$, i.e., it is close to the natural frequency of the uncoupled linear oscillator. We also consider a periodic boundary condition,
		\begin{equation}
		u_{j+N}(t) = u_{j}(t), \qquad \text{for all} j \in \mathbb{Z}_N.
		\end{equation}
		
		Considering small-amplitude oscillations, one commonly uses the rotating wave approximation
		\begin{align}\label{ansatzerror}
		u_j(t)\approx X_j(t) =& \sqrt{\epsilon } A_j\left(\tau\right) e^{i \Omega t} + \frac{1}{8} \xi \epsilon ^{3/2} 
		A_j^3\left(\tau\right) e^{3 i \Omega t} +c.c.,
		\end{align}
		i.e., $X_j(t)$ is the leading order approximation of $u_j(t)$ and $\tau = \frac{\epsilon  t}{2}$ is the slow time variable. Substituting the ansatz (\ref{ansatzerror}) into Eq.\ (\ref{dKGdriven}) and removing the resonant terms $e^{\pm i \Omega t}$ at the leading order of $\mathcal{O}(\epsilon^{3/2})$, we obtain the damped, driven discrete nonlinear Schr\"{o}dinger equation
		\begin{equation}\label{dNLSerror}
		i \dot{A}_j + 3 \xi |A_j|^2 A_j - \Delta_2 A_j + i \hat{\alpha} A_j - \hat{h} + \omega A_j = 0,
		\end{equation}
		where $\alpha = \epsilon \hat{\alpha}$, $h = 2 \epsilon^{3/2} \hat{h}$, and $A_{j+N}(t) = A_{j}(t),$ i.e., the periodic boundary condition.
		
		Using Eqs.\ (\ref{ansatzerror}) and (\ref{dNLSerror}) to approximate the solutions of (\ref{dKGdriven}) will yield the residual terms
		\begin{eqnarray} \label{res2}
		\mathrm{Res}_j(t) & := & \epsilon^{5/2} \left[ \frac{e^{i\Omega t}}{2} \left( \frac{3}{4} \xi^2 A_j^3 \bar{A}_j^2 - i \hat{\alpha } \omega  A_j + \hat{\alpha } \dot{A}_j - \frac{\omega ^2 A_j}{2} -  i \omega  \dot{A}_j + \frac{\ddot{A}_j}{2} \right) + \frac{e^{3i\Omega t}}{8} \left( 6 \xi^2 A_j^4 \bar{A}_j \right. \right. \nonumber\\
		& & \left. \qquad \left. + 3 i \hat{\alpha } \xi  A_j^3 + 9 \xi  \omega  A_j^3 + 9 i \xi  A_j^2 \dot{A}_j + 2\xi  A_j^3 - \xi  A_{j-1}^3 -\xi  A_{j+1}^3 \right) + \frac{e^{5i\Omega t}}{8} \left( 3\xi^2 A_j^5 \right) \right] \nonumber \\
		& & +\epsilon^{7/2} \left[ \frac{e^{i\Omega t}}{32} \left( 3 \xi^3 A_j^3 \bar{A}_j^4 \right) + \frac{e^{3i\Omega t}}{32} \left( -6 i \hat{\alpha } \xi  \omega  A_j^3 + 6 \hat{\alpha } \xi  A_j^2 \dot{A}_j - 9 \xi  \omega ^2 A_j^3 - 18 i \xi  \omega  A_j^2 \dot{A}_j \right. \right. \nonumber\\
		& & \left.\qquad \left.+ 6 \xi  A_j \left(\dot{A}_j\right)^2 + 3 \xi  A_j^2 \ddot{A}_j \right) + \frac{e^{5i\Omega t}}{64} \left( 3 \xi^3 A_j^6 \bar{A}_j \right) + \frac{e^{7i\Omega t}}{64} \left( 3 \xi^3 A_j^7 \right) \right] \nonumber \\
		& & + \epsilon^{9/2}\left[ \frac{e^{3i\Omega t}}{512} \left( 3 \xi^4 A_j^6 \bar{A}_j^3 \right) + \frac{e^{9i\Omega t}}{512} \left(\xi^4 A_j^9 \right) \right] + \mathrm{c.c.}
		\end{eqnarray}
		{The terms with derivatives of $A_j$ 
			can be changed into those without derivative using Eq.\ (\ref{dNLSerror}), provided that $\left({A_j}\right)_{j\in \mathbb{Z}_N}$ is a twice differentiable sequence with respect to time. Because of the periodic boundary condition, we consider the sequence space $\ell^2(\mathbb{Z}_{N})$ and we will simply denote $\left({A_j}\right)_{j\in \mathbb{Z}_N} \in \ell^2(\mathbb{Z}_N)$ by $A$. The space $\ell^2(\mathbb{Z}_N)$ is a Hilbert space equipped with norm,
			\begin{equation}
			\|A\|_{\ell^2(\mathbb{Z}_N)} = \sum_{j=1}^{N} |A_j|^2\:.
			\end{equation}}
		
		The following lemma gives us a preliminary result on the global solutions of the discrete nonlinear Schr\"{o}dinger equation (\ref{dNLSerror}).
		
		\begin{lemma}\label{Global}
			
			Let $\phi \in \ell^2(\mathbb{Z}_N)$ be an initial data. Then there exists a unique global solution $A(\tau)$ of the discrete nonlinear Schr\"{o}dinger equation (\ref{dNLSerror}) in $\ell^2(\mathbb{Z}_N)$ such that $A(0) = \phi$. Moreover, the solution $A(\tau)$ is smooth in $\tau$ and there is a real constant $C_A$, that depends on the initial data, $\hat{h}$, $\hat{\alpha}$ and $N$ such that $\lVert A(\tau)\rVert_{\ell^2(\mathbb{Z}_N)} \leq C_A$.
			
		\end{lemma}
		\begin{proof}
			In proving the global uniqueness in the lemma, we start by considering the existence of local solution. Let us rewrite Eq.\ (\ref{dNLSerror}) in the following equivalent integral form \\
			\begin{equation}\label{Intform}
			A_j(\tau) = A_j(0) - i \int_{0}^{\tau} \left(\Delta_2 A_j - i \hat{\alpha} A_j - 3\xi |A_j|^2 A_j - \omega A_j + \hat{h}\right) ds\:.
			\end{equation}
			Define a closed ball of radius $\delta$, that needs not necessarily be small, in $\ell^2(\mathbb{Z}_N)$, 
			\begin{equation}
			\mathcal{B} = \{A \in C\left([0,\tau_m],\ell^2(\mathbb{Z}_N)\right)\;|\; \lVert A (\tau) \rVert_{\ell^2(\mathbb{Z}_N)}  \leq \delta \}\:,
			\end{equation}
			equipped with the norm 
			\begin{equation*}
			\lVert A \rVert_{\mathcal{B}} = \sup_{\tau \in [0,T]} \lVert A (\tau) \rVert_{\ell^2(\mathbb{Z}_N)}.
			\end{equation*}
			
			For $A \in \ell^2(\mathbb{Z}_N)$, we define a nonlinear operator,
			\begin{equation}\label{operatorK}
			K_j\left[A(\tau)\right] = \phi - i \int_{0}^{\tau} \left(\Delta_2 A_j - i \hat{\alpha} A_j - 3\xi |A_j|^2 A_j - \omega A_j + \hat{h}\right) ds.
			\end{equation}
			We want to prove that $K$ is contraction mapping on $\mathcal{B}$.
			
			Due to the periodic boundary condition, $A_j = A_{N+j}$, we get
			\begin{equation}
			\lVert \Delta_2 A \rVert_{\ell^2(\mathbb{Z}_N)}\leq C_{\Delta_2} \lVert A \rVert_{\ell^2(\mathbb{Z}_N)}.
			\end{equation}
			Therefore, the discrete Laplacian $\Delta_2$ is a bounded operator in $\ell^2(\mathbb{Z}_N)$. From the Banach algebra property of the $\ell^2(\mathbb{Z}_N)$-space, there is a constant $C > 0$ such that for every $A, B \in \ell^2(\mathbb{Z}_N)$ we have
			\begin{equation}\label{Banach}
			\lVert A B\rVert_{\ell^2(\mathbb{Z}_N)} \leq C \lVert A \rVert_{\ell^2(\mathbb{Z}_N)} \lVert B \rVert_{\ell^2(\mathbb{Z}_N)}.
			\end{equation}
			From (\ref{operatorK}) and using the estimate (\ref{Banach}) we obtain the following bound 
			\begin{equation*}
			\lVert K (A)\rVert_{\mathcal{B}} \leq \delta_0 + T(C\delta + \alpha \delta + 3C\delta^3 + \omega \delta).
			\end{equation*}
			We can pick $\delta_0 < \frac{\delta}{2}$ and $T \leq \frac{\delta}{2(C\delta + \alpha \delta + 3C\delta^3 + \omega \delta)}$. Thus, $K$ is a mapping from $\overline{B}_\delta$ to itself.\\
			
			For $A,B \in \overline{B}_\delta$, we have 
			\begin{align*}
			K_j[A(\tau)] - K_j[B(\tau)] = -i \int_{0}^{\tau}& \left[\Delta_2(A_j-B_j) - i\hat{\alpha} (A_j-B_j) - 3\xi(|A_j|^2A_j - |B_j|^2 B_j)\right.\\
			&\left.- \omega (A_j-B_j)\right] ds.
			\end{align*}	
			Therefore, we obtain 
			\begin{align*}
			\lVert K(A) - K(B) \rVert_{\mathcal{B}}\leq T \{C_\Delta + \hat{\alpha} + 3C\hat{\alpha}^2 + \omega\}\lVert A-B \rVert_{\mathcal{B}} .
			\end{align*}
			By taking $T < \mathrm{min}\left(\frac{1}{C_\Delta + \hat{\alpha} + 3C\hat{\alpha}^2 + \omega}, \frac{\delta}{2(C\delta + \alpha \delta + 3C\delta^3 + \omega \delta)}\right)$, then $K$ is a contraction mapping.\\
			{Therefore, there exists a constant $T$ such that the discrete nonlinear Schr\"{o}dinger equation has a local unique solution $A \in \mathcal{C}\left([0,T], \ell^2(\mathbb{Z}_N)\right)$ with $\sup_{\tau \in [0,T]} \lVert A (\tau) \rVert_{\ell^2(\mathbb{Z}_N)} \leq \delta$.} We can construct the maximal solution by repeating the process above and gluing the solution using uniqueness condition.
			
			Now, we will prove the global well-posedness of the discrete nonlinear Schr\"{o}dinger equation (\ref{dNLSerror}). Multiplying the $j$th-component of the equation by $\bar{A}_j$, taking the imaginary part and summing over $j$, we obtain
			\begin{equation}
			\dfrac{d}{d\tau} \lVert A \rVert^2_{\ell^2(\mathbb{Z}_N)} + 2   \hat{\alpha}\lVert A \rVert^2_{\ell^2(\mathbb{Z}_N)} =  2 |\hat{h}| \sum_{j=1}^{N} \mathrm{Im}(\bar{A}_j) \leq 2  |\hat{h}| \lVert A \rVert_{\ell^2(\mathbb{Z}_N)} \;.
			\end{equation}
			{Integrating the inequality, we get
				\begin{eqnarray}
				\lVert A \rVert^2_{\ell^2(\mathbb{Z}_N)} & \leq & \dfrac{|\hat{h}|}{\hat{\alpha}} + \left(\lVert \phi \rVert_{\ell^2(\mathbb{Z}_N)} -\dfrac{|\hat{h}|}{\hat{\alpha}} \right) e^{-\hat{\alpha}\tau} \nonumber\\
				& \leq & C_A\left(\lVert \phi\rVert_{\ell^2(\mathbb{Z}_N)}, \hat{h}, \hat{\alpha},N\right)\:,
				\end{eqnarray}
				which provides a global bound to the solutions and hence, conclude the proof of the lemma.}
		\end{proof}
		
		The following lemma will give us an estimate for the leading order approximation (\ref{ansatzerror}).
		\begin{lemma}\label{ConstantX}
			{For every $A_0 \in \ell^2(\mathbb{Z}_N)$, there exits a positive constant $C_X(\lVert A_0\rVert_{\ell^2(\mathbb{Z}_N)}, \hat{h}, \hat{\alpha},N)$ such that the leading-order approximation (\ref{ansatzerror}) is estimated by 
				\begin{equation}
				\lVert X(t)\rVert_{\ell^2(\mathbb{Z}_N)} + 	\lVert \dot{X}(t) \rVert_{\ell^2(\mathbb{Z}_N)} \leq \sqrt{\epsilon}\; C_X (\lVert A_0 \rVert_{\ell^2(\mathbb{Z}_N)}, \hat{h}, \hat{\alpha},N)\;,
				\end{equation}}
			for all $t\in[0,\infty)$ and $\epsilon \in (0,1)$.
		\end{lemma}	
		
		\begin{proof}
			{From the global existence in Lemma \ref{Global} and using the Banach algebra property of $\ell^2(\mathbb{Z}_N)$, we obtain 
				\begin{align}\label{X(t)}
				\begin{aligned}
				\left\lVert X(t)\right\rVert_{\ell^2(\mathbb{Z}_N)} & = \left\lVert\sqrt{\epsilon } \left(A e^{i \Omega t}  + \bar{A} e^{-i \Omega t}\right) + \frac{1}{8} \xi \epsilon ^{3/2} \left(A^3 e^{3 i \Omega t} + \bar{A}^3 e^{-3 i \Omega t}\right)\right\rVert_{\ell^2(\mathbb{Z}_N)}\\
				& \leq \left\lVert\sqrt{\epsilon } \left(A e^{i \Omega t}  + \bar{A} e^{-i \Omega t}\right) \right\rVert_{\ell^2(\mathbb{Z}_N)} + \left\lVert \frac{1}{8} \xi \epsilon ^{3/2} \left(A^3 e^{3 i \Omega t} + \bar{A}^3 e^{-3 i \Omega t}\right)\right\rVert_{\ell^2(\mathbb{Z}_N)}\\
				& \leq 2 \sqrt{\epsilon } \left\lVert A\right\rVert_{\ell^2(\mathbb{Z}_N)} + \frac{1}{4} \xi \epsilon ^{3/2} \left\lVert A^3\right\rVert_{\ell^2(\mathbb{Z}_N)} \\
				& \leq \sqrt{\epsilon}\;C_{X_1} (\lVert A_0 \rVert_{\ell^2(\mathbb{Z}_N)}, \hat{h}, \hat{\alpha},N)
				\end{aligned}
				\end{align}
				and 
				\begin{align}\label{Xdot}
				\begin{aligned}
				\left\lVert\dot{X}(t)\right\rVert_{\ell^2(\mathbb{Z}_N)} = & 	\left\lVert\frac{1}{8} \xi \epsilon ^{3/2} \left(\frac{3}{2} \epsilon  \bar{A}^2 e^{-3 i \Omega t
				} \dot{\bar{A}} -3 i \Omega \bar{A}^3 e^{-3 i \Omega t } + \frac{3}{2}
				\epsilon  A^2 \dot{A} e^{3 i \Omega t } +3 i \Omega A^3 e^{3 i \Omega t }\right) \right.\\
				& \left. \quad+\sqrt{\epsilon } \left(\frac{1}{2} \epsilon  e^{-i\Omega t
				} \dot{\bar{A}} - i \Omega \bar{A} e^{-i\Omega t } +\frac{1}{2} \epsilon \dot{A} e^{i \Omega t }+i \Omega A  e^{i\Omega t }\right) \right\rVert_{\ell^2(\mathbb{Z}_N)} .
				\end{aligned} 
				\end{align}
				From (\ref{dNLSerror}),
				we have that
				\begin{align}\label{Adot2}
				\begin{aligned}
				\left\lVert \dot{A}(\tau) \right\rVert_{\ell^2(\mathbb{Z}_N)} & = \left\lVert 3 i \xi A^2 \bar{A} -\hat{\alpha}  A + i \omega  A - i(\Delta_2 A) - i \hat{h} \right\rVert_{\ell^2(\mathbb{Z}_N)}\\
				&\leq 3 \xi C_A^3\left(\lVert A_0\rVert_{\ell^2(\mathbb{Z}_N)}, \hat{h}, \hat{\alpha}\right) + \alpha C_A\left(\lVert A_0\rVert_{\ell^2(\mathbb{Z}_N)}, \hat{h}, \hat{\alpha}\right) + \omega C_A\left(\lVert A_0\rVert_{\ell^2(\mathbb{Z}_N)}, \hat{h}, \hat{\alpha}\right) \\
				&\quad+ C_{\Delta_2} C_A\left(\lVert A_0\rVert_{\ell^2(\mathbb{Z}_N)}, \hat{h}, \hat{\alpha}\right) + C \hat{h} \\
				& \leq \tilde{C}_A\left(\lVert A_0\rVert_{\ell^2(\mathbb{Z}_N)}, \hat{h}, \hat{\alpha},N\right).
				\end{aligned}
				\end{align}
				Therefore, Eq. (\ref{Xdot}) becomes
				\begin{align}\label{Xdot2}
				\left\lVert\dot{X}(t)\right\rVert_{\ell^2(\mathbb{Z}_N)} \leq \sqrt{\epsilon}\;C_{X_2} (\lVert A_0 \rVert_{\ell^2(\mathbb{Z}_N),}, \hat{h}, \hat{\alpha},N)
				\end{align}
				and
				\begin{align*}
				\lVert X(t)\rVert_{\ell^2(\mathbb{Z}_N)} + 	\lVert \dot{X}(t) \rVert_{\ell^2(\mathbb{Z}_N)}   \leq \sqrt{\epsilon}\;C_X (\lVert A_0 \rVert_{\ell^2(\mathbb{Z}_N)}, \hat{h}, \hat{\alpha},N)
				\end{align*} }
		\end{proof}
		
		Next, we have the following result on the bound of the residual terms Eq. (\ref{res2}).
		\begin{lemma}\label{ConstantR}
			{For every $A_0 \in \ell^2(\mathbb{Z}_N)$, there exists a positive $\epsilon-$independent constant\\ $C_R(\lVert A_0\rVert_{\ell^2(\mathbb{Z}_N)}, \hat{h}, \hat{\alpha},N)$, such that for every $\epsilon \in (0,1)$ and every $ t \in \mathbb{R}$, the residual term in (\ref{res2}) is estimated by
				\begin{equation}\label{Residu}
				\lVert \mathrm{Res}(t) \rVert_{\ell^2(\mathbb{Z}_N)} \leq C_R \left(\lVert A_0 \rVert_{\ell^2(\mathbb{Z}_N)}, \hat{h}, \hat{\alpha},N\right)\epsilon^{5/2}.
				\end{equation}}
		\end{lemma}	
		
		\begin{proof}
			To prove this lemma, we can use the result from Lemma \ref{Global} as well as the property of Banach algebra in $\ell^2(\mathbb{Z}_N)$, such that from the global existence and smoothness of the solution $A(\tau)$ of the discrete nonlinear Schr\"odinger equation (\ref{dNLSerror}) in Lemma \ref{Global}, we obtain (\ref{Residu}).
		\end{proof}
		
		\section{Main Results}
		\label{sec3}
		
		In this section we will develop the main result on the time evolution of the rotating-wave approximation error by writing $u_j(t)= X_j(t) + y_j(t)$, where $X_j(t)$ is the leading-order approximation (\ref{ansatzerror}) and $y_j(t)$ is the error term. Plugging the decomposition into Eq.\ (\ref{dKGdriven}), we obtain the evolution problem for the error term:
		\begin{align}\label{erroreq}
		\ddot{y}_j + y_j + \xi \left(y_j^3 + 3 X_j^2 y_j + 3 X_j y_j^2\right) - \epsilon \Delta_2 y_j + \epsilon \hat{\alpha} \dot{y}_j + \mathrm{Res}_j(t) = 0, \quad  j\in\mathbb{Z}_N,
		\end{align}
		where the residual term $\mathrm{Res}_j(t)$ is given by (\ref{res2}) if $A(\tau)$ satisfies Eq.\ (\ref{dNLSerror}). {Since $u$ and $X$ satisfy periodic boundary conditions, the error term $y$ also satisfies the same condition $y_{j+N} (t) = y_{j} (t)$.
			
			Associated with Eq.\ (\ref{erroreq}), we can define the energy of the error term as
			\begin{equation}\label{energy}
			E(t) := \frac{1}{2} \sum_{j=1}^{N} \left[\dot{y}_j^2 + y_j^2 - 2\epsilon \left(y_j y_{j+1} - y_j^2\right)\right].
			\end{equation}
			For every $t$ for which the solution $y(t)$ is defined,  we have $E(t) \geq 0$ and the following inequality,
			\begin{equation}\label{normsol}
			\lVert \dot{y}(t) \rVert_{\ell^2(\mathbb{Z}_N)}^2 + \lVert y(t) \rVert_{\ell^2(\mathbb{Z}_N)}^2 \leq 2 E(t).
			\end{equation}
			The rate of change for the energy (\ref{energy}) is found from the evolution problem (\ref{erroreq}) as follows
			\begin{equation}\label{rate}
			\frac{dE}{dt} = -\sum_{j = 1}^{N} \left[ \mathrm{Res}_j(t) + \epsilon \hat{\alpha} \dot{y}_j + \xi (y_j^3 + 3 X_j^2 y_j + 3 X_j y_j^2)\right] \dot{y}_j.
			\end{equation}
			
			Using the Cauchy-Schwarz inequality and setting $E = Q^2$, we get 
			\begin{eqnarray}
			\left|\frac{dQ}{dt}\right| & \leq & \frac{1}{\sqrt{2}}\lVert \mathrm{Res}(t)\rVert_{\ell^2(\mathbb{Z}_N)} + \left[\epsilon\hat{\alpha} + 2|\xi|Q^2 + 3|\xi|\sqrt{2}\lVert X(t)\rVert_{\ell^2(\mathbb{Z}_N)} Q \right. \nonumber\\
			& & \qquad \left.+ 3|\xi|\sqrt{2}\lVert X(t)\rVert^2_{\ell^2(\mathbb{Z}_N)}\right]Q.
			\end{eqnarray}
			
			Take $\tau_0 > 0$ arbitrarily. Assume that the initial norm of the perturbation term satisfies the following bound
			\begin{equation}\label{initialQ}
			Q(0) \leq C_0 \epsilon^{3/2},
			\end{equation}
			where $C_0$ is a positive constant, and define 
			\begin{equation}\label{T0}
			T_0 = \sup \left \{t_0 \in [0,2\tau_0\epsilon^{-1}] : \sup_{t \in [0,t_0]} Q(t) \leq C_Q \epsilon^{3/2} \right \},
			\end{equation}
			on the time scale $[0, 2\tau_0\epsilon^{-1}]$.
			
			Applying Lemmas \ref{ConstantX}-\ref{ConstantR} and the definition (\ref{T0}), we have
			\begin{equation}
			\left|\frac{dQ}{dt}\right| \leq \frac{\epsilon^{5/2}C_R}{\sqrt{2}} + \left(2\hat{\alpha} + 4|\xi| C_Q^2\epsilon^2 + 6|\xi|\sqrt{2}C_Q C_X \epsilon + 6 |\xi|C_X^2 \right) \frac{\epsilon Q}{2} \:. \label{estimateE}
			\end{equation}
			Thus, for every $t \in [0,T_0]$ and $\epsilon > 0$ which is sufficiently small, we can find a positive constant $K_0$, which is independent of $\epsilon$, such that 
			\begin{equation}
			2\hat{\alpha} + 4|\xi| C_Q^2\epsilon^2 + 6|\xi|\sqrt{2}C_Q C_X \epsilon + 6 |\xi|C_X^2 \leq K_0 .
			\end{equation}
			Integrating (\ref{estimateE}), we get
			\begin{equation}
			Q(t) e^{-\frac{\epsilon K_0 t}{2}} - Q(0) \leq \int_{0}^{t} \frac{C_R \epsilon^{5/2}}{\sqrt{2}}  e^{-\frac{\epsilon K_0 s}{2}} ds \leq \frac{\sqrt{2}C_R \epsilon^{3/2}}{K_0}.
			\end{equation}
			Since we assume (\ref{initialQ}) holds, then we obtain
			\begin{equation}
			Q(t) \leq \epsilon^{3/2} \left(C_0 + \frac{\sqrt{2}C_R}{ K_0}\right) e^{K_0 \tau_0}.
			\end{equation}
			Therefore, we can define $C_Q:= \left(C_0 + 2^{1/2} K_0^{-1} C_R\right) e^{K_0 \tau_0}$.  
			
			Based on the above analysis, we can state the main result of this paper in the following theorem.
			\begin{theorem}\label{theorem1}
				{For every $\tau_0 > 0$, there are a small $\epsilon_0 > 0$ and positive constants $C_0$ and $C$ such that for every $\epsilon \in (0,\epsilon_0)$, for which the initial data satisfies 
					\begin{equation}\label{data1}
					\lVert y(0) \rVert_{\ell^2(\mathbb{Z}_N)} + \lVert \dot{y}(0) \rVert_{\ell^2(\mathbb{Z}_N)} \leq C_0 \epsilon^{3/2} ,
					\end{equation}
					the solution of the discrete Klein-Gordon equation (\ref{dKGdriven}) satisfies for every $t \in [0, 2\tau_0\epsilon^{-1}]$,
					\begin{equation}\label{data2}
					\lVert y(t) \rVert_{\ell^2(\mathbb{Z}_N)} + \lVert \dot{y}(t) \rVert_{\ell^2(\mathbb{Z}_N)} \leq C \epsilon^{3/2} .
					\end{equation}}
			\end{theorem}
			
			\begin{remark}
				The error bound that is of order $\mathcal{O}(\epsilon^{3/2})$ in Theorem \ref{theorem1} is linked to the choice of our rotating wave ansatz \eqref{ansatzerror} that creates a residue of order $\mathcal{O}(\epsilon^{5/2})$. If we include a higher-order correction term in the ansatz \eqref{ansatzerror}, see Chapter 5.3 of \cite{9} for the procedure to do it, we will obtain a smaller residue and in return a smaller error bound. 
			\end{remark}
			
			\section{Numerical Discussions} 
			\label{sec4}
			
			\begin{figure}[tbhp!]
				\centering
				\subfigure[]{\includegraphics[scale=0.4]{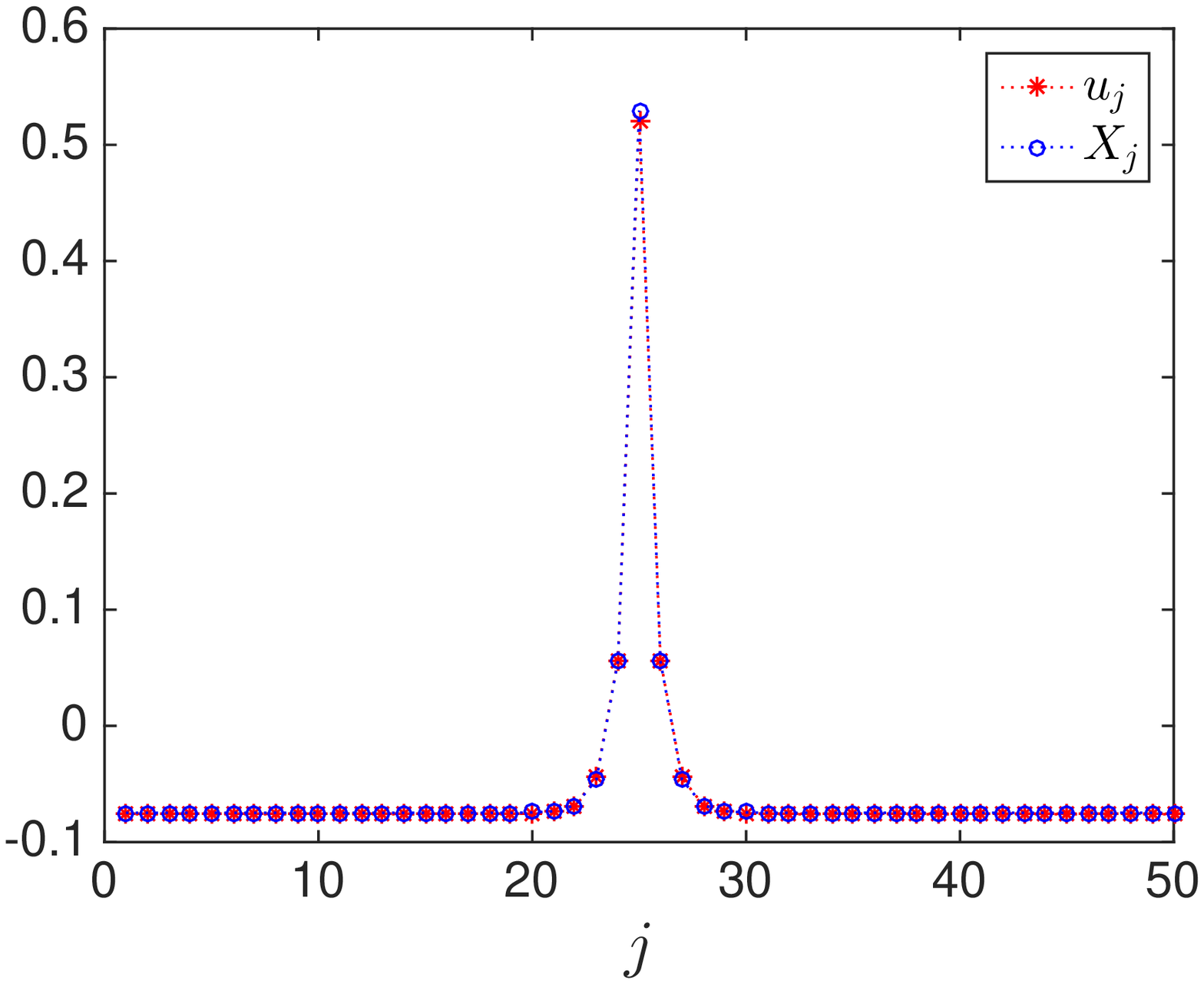}\label{fig1a}}
				\subfigure[]{\includegraphics[scale=0.4]{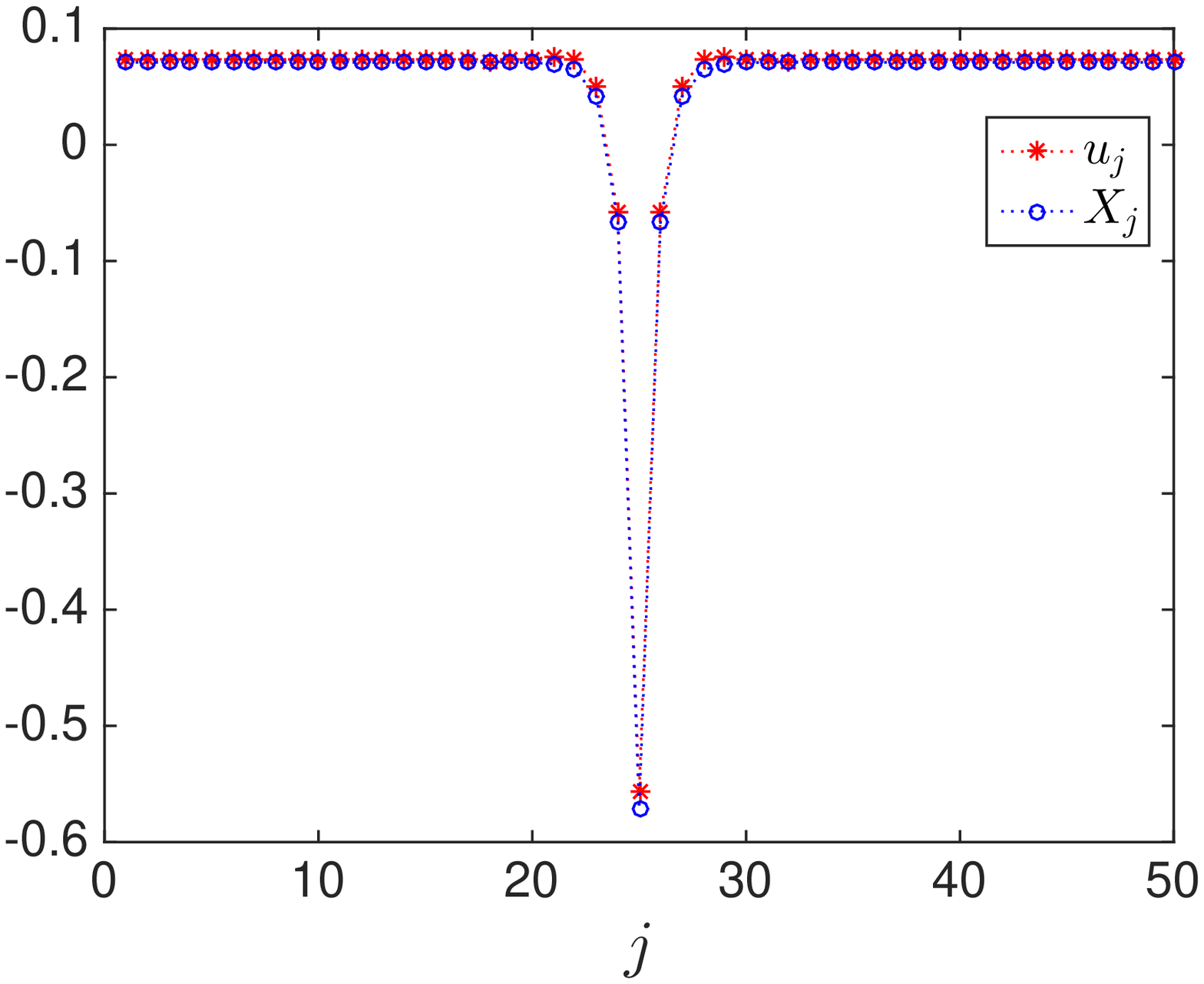}\label{fig1b}}
				\subfigure[]{\includegraphics[scale=0.4]{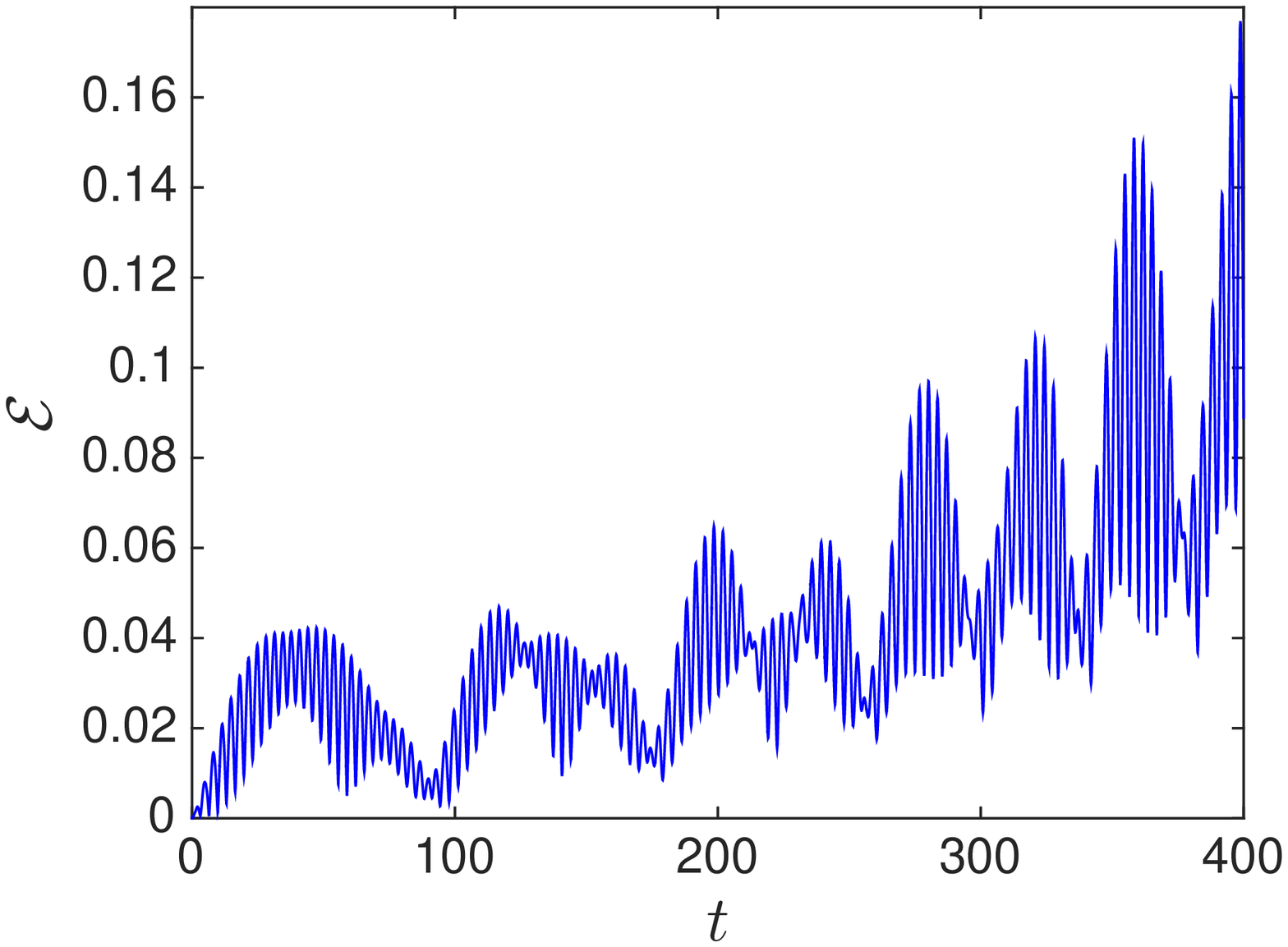}\label{fig1c}}	
				\subfigure[]{\includegraphics[scale=0.4]{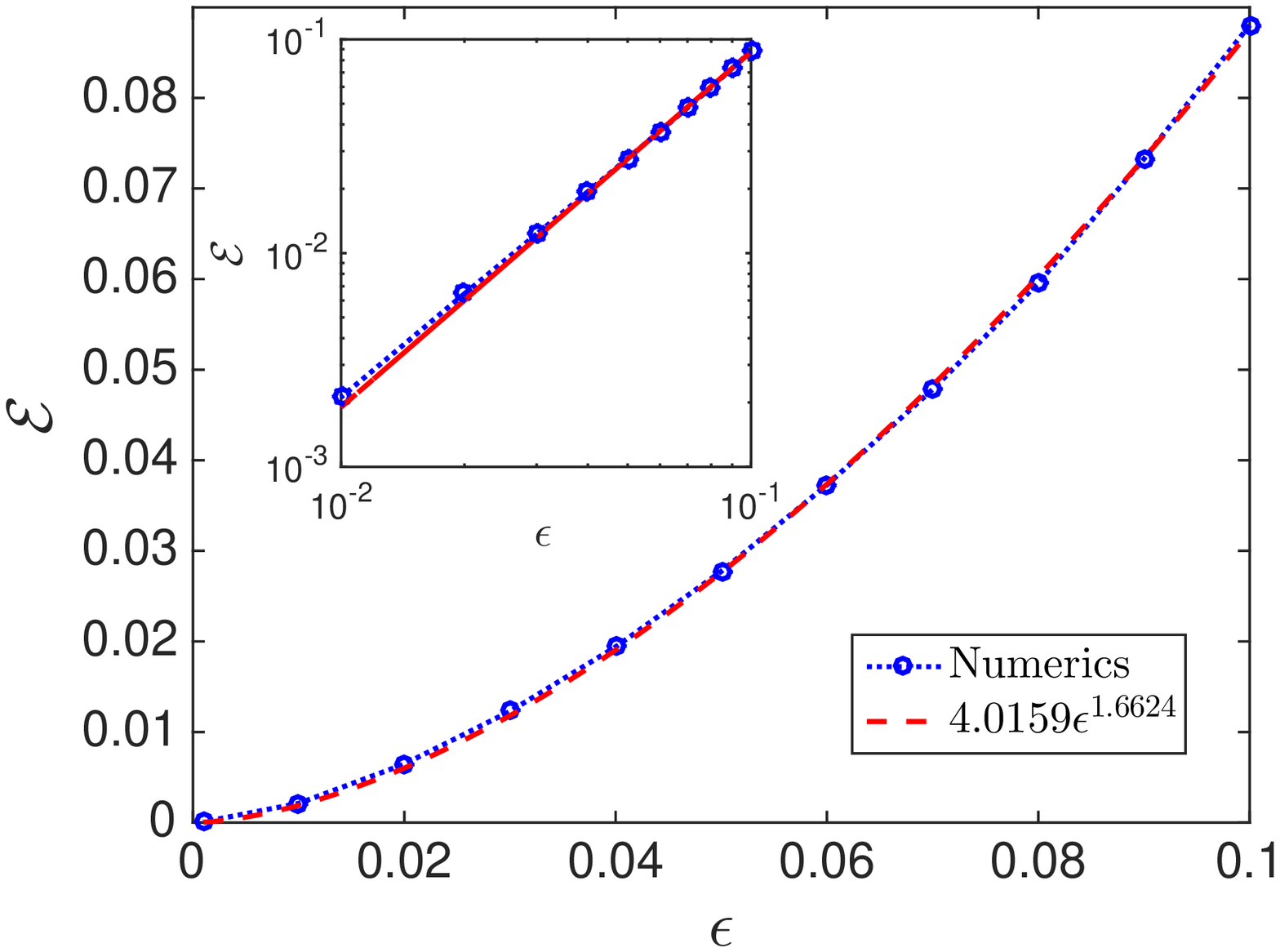}\label{fig1d}}
				\caption{(a,b) Numerical solutions of the Klein-Gordon equation (blue circles) and the corresponding rotating wave approximations from the Schr\"odinger equation (red stars) at two time instances $t=75$ and $t=200$. Here, $\epsilon=0.05$. (c) Time dynamics of the error. (d) Maximum error of the Schr\"odinger approximation within the interval $t\in[0,2/\epsilon]$ for varying $\epsilon\to0$. In the picture, we also plot the best power fit of the error, showing that the error approximately has the same order as in Theorem \ref{theorem1}.
				}
				\label{fig1}
			\end{figure}

			\begin{figure}[htbp]
				\centering
				\subfigure[]{\includegraphics[scale=0.5]{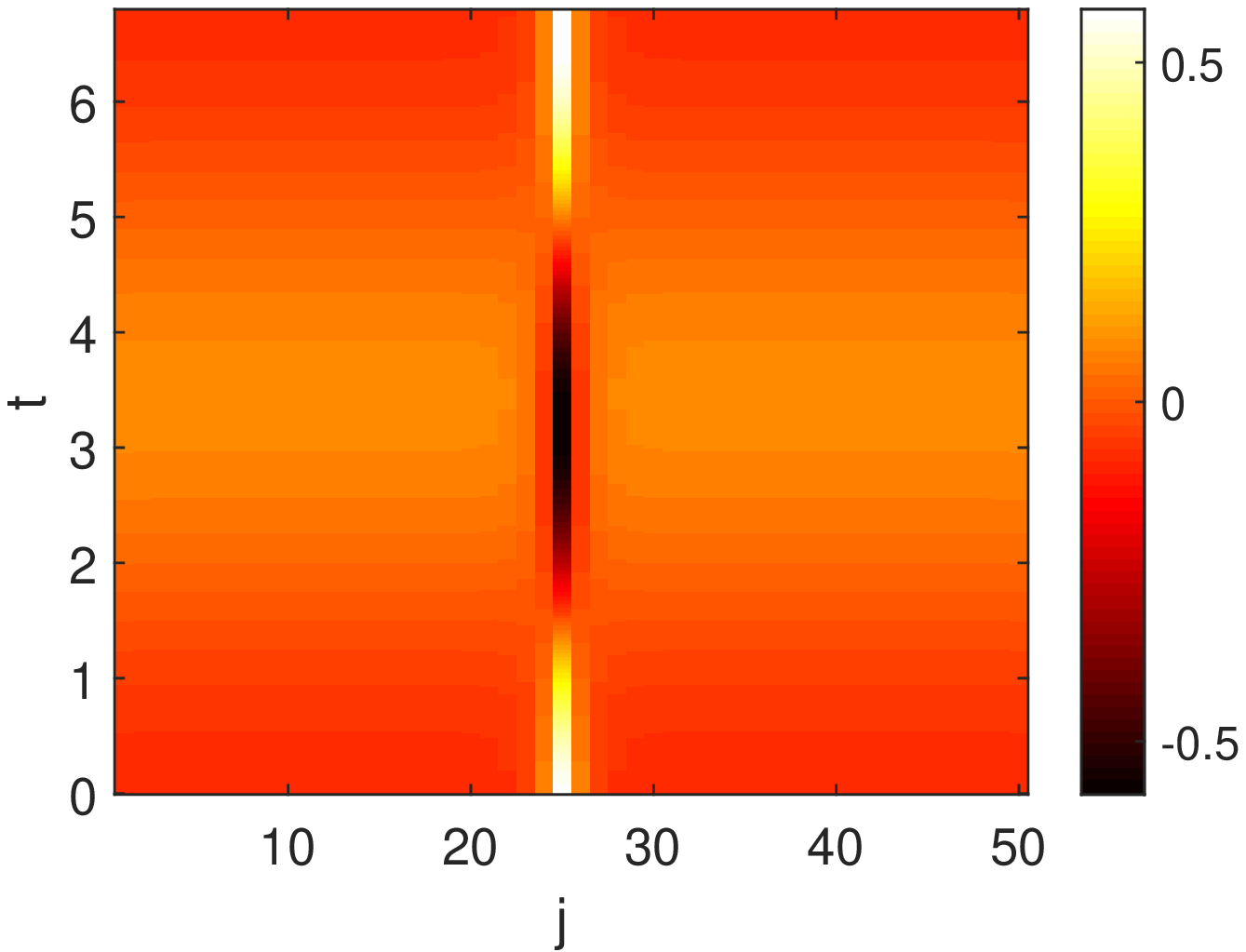}\label{subfig:prof_kg_eps_0_05}}
				\subfigure[]{\includegraphics[scale=0.5]{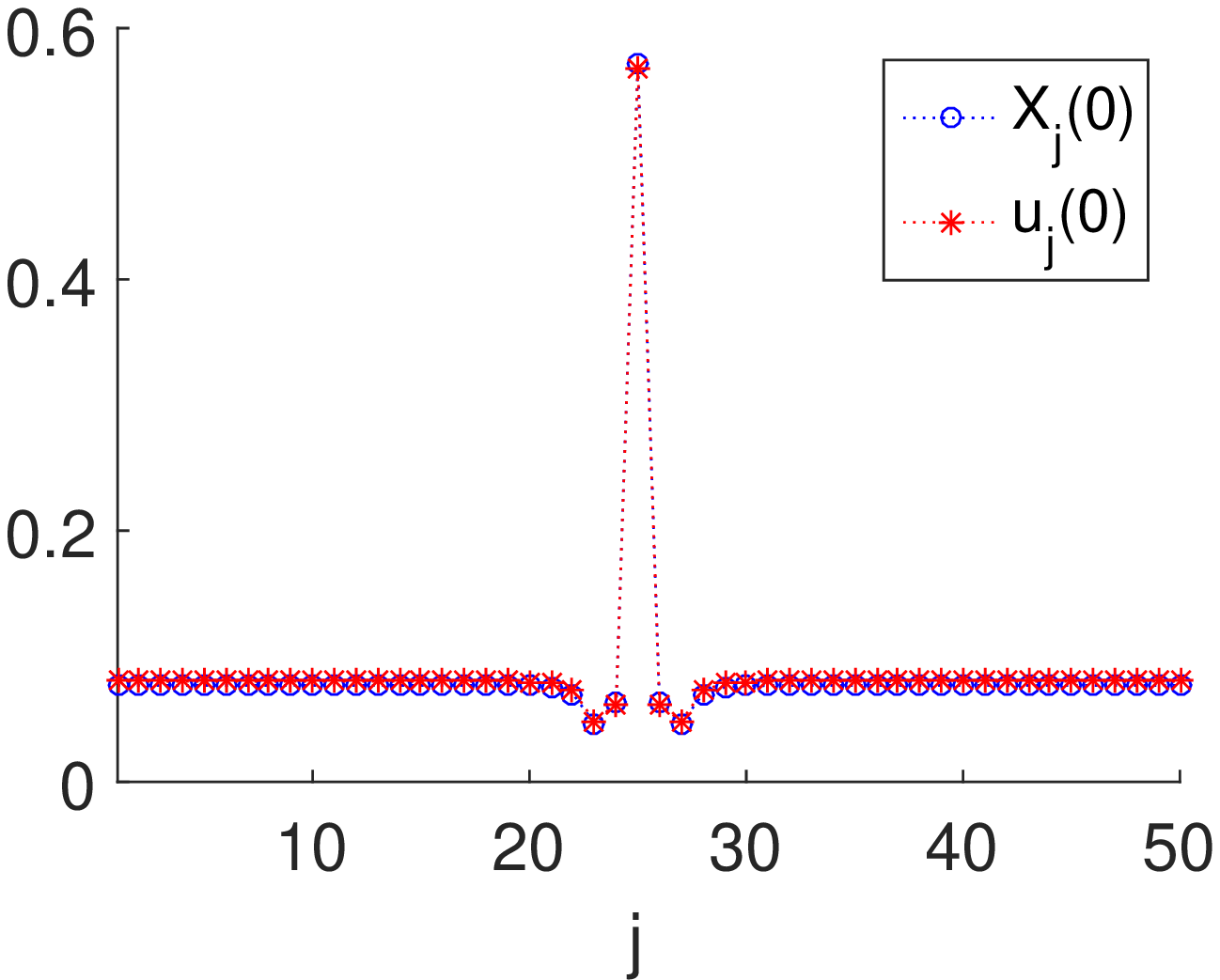}\label{subfig:prof_compare_kg_dnls_eps_0_05}}
				\caption{Breather solution of \eqref{dKGdriven} for ${\epsilon}=0.05$.
					Panel (a) shows the dynamics of the solution in one period, while panel (b) presents the comparison of the breather and its approximation \eqref{ansatzerror}, with $A_j$ obtained from solving the Schr\"odinger equation \eqref{dNLSerror}.
				}
				\label{fig:onsite_stability_compare_kg_dnls}
			\end{figure}
			
			We have discussed in Section \ref{sec2}, that as an approximate solution of the Klein-Gordon equation \eqref{dKGdriven}, the rotating wave approximation \eqref{ansatzerror} and \eqref{dNLSerror} yields a residue of order $\mathcal{O}\left(\epsilon^{5/2}\right)$. Following on the result, we proved in Section \ref{sec3} that the difference between solutions of Eqs.\ \eqref{dKGdriven} and \eqref{dNLSerror} that are initially of at most order $\mathcal{O}\left(\epsilon^{3/2}\right)$ will be of the same order for some finite time $2\tau_0/\epsilon$, for $\tau_0>0$. In this section, we will illustrate the analytical result on the error bound above numerically.  
			
			\subsection{Error growth}
			
			We consider Eq.\ \eqref{dKGdriven} as an initial value problem, that is then integrated using the fourth-order Runge-Kutta method. To compare solutions of Eq.\ \eqref{dKGdriven} and the rotating wave approximation \eqref{ansatzerror}, simultaneously we also need to integrate Eq.\ \eqref{dNLSerror}. As the initial data of the Klein-Gordon equation, we take 
			\begin{eqnarray*}
				u_j(0)&=& \left.{\epsilon^{1/2} } A_j+ \frac{1}{8} \xi \epsilon ^{3/2} 
				A_j^3\right|_{t=0}+c.c.,\\
				\dot{u}_j(0)&=&\left. {\epsilon^{1/2} }\left[ {A_j}_\tau+i \Omega {A_j} \right] + \frac{1}{8} \xi \epsilon ^{3/2} A_j^2\left[
				3{A_j}_\tau+3 i \Omega A_j\right]\right|_{t=0} +c.c.,
			\end{eqnarray*}
			where ${A_j}_\tau(0)$ can be obtained from the Schr\"odinger equation \eqref{dNLSerror}. In this way, the error $y(t)=u_j(t)-X_j(t)$ will satisfy the initial condition $\lVert y(0)\rVert_{\ell^2}=0$. In the following, we take the parameter values $\omega=3$, $\hat{h}=-0.5$, $\hat{\alpha}=0.1$,  and the nonlinearity coefficient $\xi=-1$. We also take the number of sites $N=50$. 
			
			For our illustration, we consider a discrete soliton, i.e., a special standing wave solution of the Schr\"odinger equation \eqref{dNLSerror} that is localised in space. Such a solution can be obtained rather immediately from solving the time-independent equation of \eqref{dNLSerror} using, e.g., Newton's method. 
			
			In Fig.\ \ref{fig1a} and \ref{fig1b} we plot the solutions $u_j(t)$ and $X_j(t)$ for $\epsilon=0.05$ at two different subsequent times. In panel (c) of the same figure, we plot the error $\lVert y(t)\rVert$ between the two solutions, which shows that it increases. However, the increment is bounded within the prediction $\sim C\epsilon^{3/2}$ for quite a long while.
			
			We have performed similar computations for several different values of $\epsilon\to0$. Taking $\tau_0=1$, we record sup$_{t\in[0,2\tau_0/\epsilon]}\lVert y(t)\rVert$ for each $\epsilon$. We plot in Fig.\ \ref{fig1d} the maximum error within the time interval as a function of $\epsilon$. We also plot in the same panel the best power fit in the nonlinear least squares sense
			, showing that the error is approximately of order $\mathcal{O}(\epsilon^{3/2})$ in agreement with Theorem \ref{theorem1}.
			
			\subsection{Discrete solitons vs.\ discrete breathers}
			
			Our simulations in Fig.\ \ref{fig1} indicate that discrete solitons of the Schr\"odinger equation shall approximate breathers, i.e., solutions that are periodic in time but localised in space, of the discrete Klein-Gordon equation. Yet, how close are the actual discrete breathers from the solitons? If they are quite close, do they share the same stability characteristics? 
			
			\begin{figure}[htbp]
				\centering
				{\includegraphics[scale=0.5]{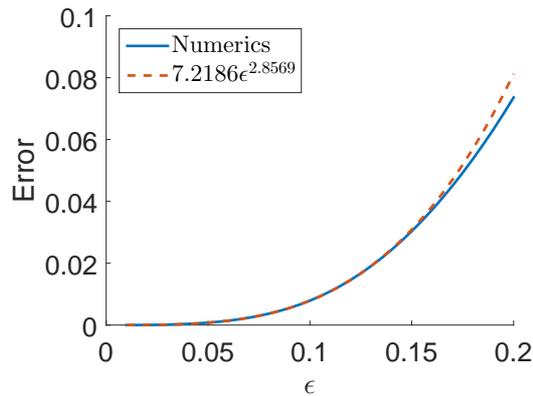}\label{subfig:error_plot}}
				\caption{Plot of the estimated error of the discrete Schr\"odinger approximation (\ref{dNLSerror}) for various $\epsilon\to0$. The dashed line is the best power fit, indicated in the legend. }
				\label{fig:err}
			\end{figure}
			
			To answer the questions, we need to look for breathers of \eqref{dKGdriven}. Due to the temporal periodicity of the solutions, we can write  $u_j(t)$ in trigonometric series:
			\begin{equation}
			u_j(t)=\sum_{k=1}^{K}a_{j,k}\cos\left(\left(k-1\right)\Omega t\right)+b_{j,k}\sin\left(k\Omega t\right),\,j=1,2,\dots,N,
			\label{eq:anz}
			\end{equation}
			where $a_{j,k}$ and $b_{j,k}$ are the Fourier coefficients and $K$ is the number of Fourier modes we will use in our numerics. Herein, we use $K=3$ and $N=50$, even though larger numbers have been used as well to make sure that the results are independent of the lattice size and the number of modes. 
			
			Substituting the series \eqref{eq:anz} into Eq.\ \eqref{dKGdriven} and integrating the resulting equation over the time-period $2\pi/\Omega$, one will obtain coupled nonlinear equations for the coefficients $a_{j,k}$ and $b_{j,k}$. We then use Newton's method to solve the resulting equations. Breathers will be obtained by properly choosing the initial guess for the coefficients.

			Once a solution, e.g., $\hat{u}_j(t)$, is obtained, we determine its linear stability using Floquet theory. Defining $u_j(t) = \hat{u}_j(t)+\delta Y_j(t),$ substituting it into Eq.\ \eqref{dKGdriven}, and linearising about $\delta=0$, we obtain the linear second-order differential-difference equation 
			\begin{equation}
			\begin{array}{ccl}
			\dot{Y}_j&=&Z_j,\\
			\dot{Z}_j&=&-Y_j-3\xi \hat{u}_j^2 Y_j+\epsilon\Delta_2Y_j-\alpha {Z}_j.
			\end{array}
			\label{fm}
			\end{equation}
			By integrating the system of linear equations until $t={2\pi}/{\Omega}$, and using a standard basis in $\mathbb{R}^{2N}$, i.e., $\left\{e^0_1,e^0_2,...,e^0_{2N}\right\}$ as the initial condition at $t=0$, we obtain a collection of solutions at $t={2\pi}/{\Omega}$:
			\begin{equation}
			M=\left\{
			E_1,E_2,...,E_{2N}
			\right\}
			\in \mathbb{R}^{2N\times2N},
			\end{equation}
			as a monodromy matrix. The solution $\hat{u}_j(t)$ is said to be linearly stable when all the eigenvalues $\lambda$ of the monodromy matrix lies inside or on the unit circle and unstable when there exists at least one $\lambda$ that is outside the unit circle. 
			
			As for discrete solitons of the Schr\"odinger equation \eqref{dNLSerror}, after a standing wave solution $\tilde{A}_j=(\tilde{x}_j+i\tilde{y}_j)$ is obtained, its linear stability can also be determined from solving the linear eigenvalue problem
			\begin{equation}
			\lambda\left(
			\begin{array}{c}
			\hat{x}_j\\
			\hat{y}_j
			\end{array}
			\right)=
			\left(
			\begin{array}{cc}
			-6\xi{x}_j{y}_j-\alpha& \Delta-\omega-3\xi\left({x}_j^2+3{y}_j^2\right)\\
			\omega-\Delta+3\xi\left(3{x}_j^2+{y}_j^2\right)&6\xi{x}_j{y}_j-\alpha
			\end{array}
			\right)
			\left(
			\begin{array}{c}
			\hat{x}_j\\
			\hat{y}_j
			\end{array}
			\right),
			\label{em}
			\end{equation}
			that is derived straightforwardly as above from substituting $A_j=\tilde{A}_j+\delta(\hat{x}_j+i\hat{y}_j)e^{\lambda \tau}$ into Eq.\ \eqref{dNLSerror} and linearising the equation about $\delta=0$. Solution $\tilde{A}_j$ is said to be linearly stable when all of the eigenvalues have Re$(\lambda)\leq0$ and unstable when there is an eigenvalue with Re$(\lambda)>0$.
			
			\begin{figure}[tbhp!]
				\centering
				\subfigure[]{\includegraphics[scale=0.5]{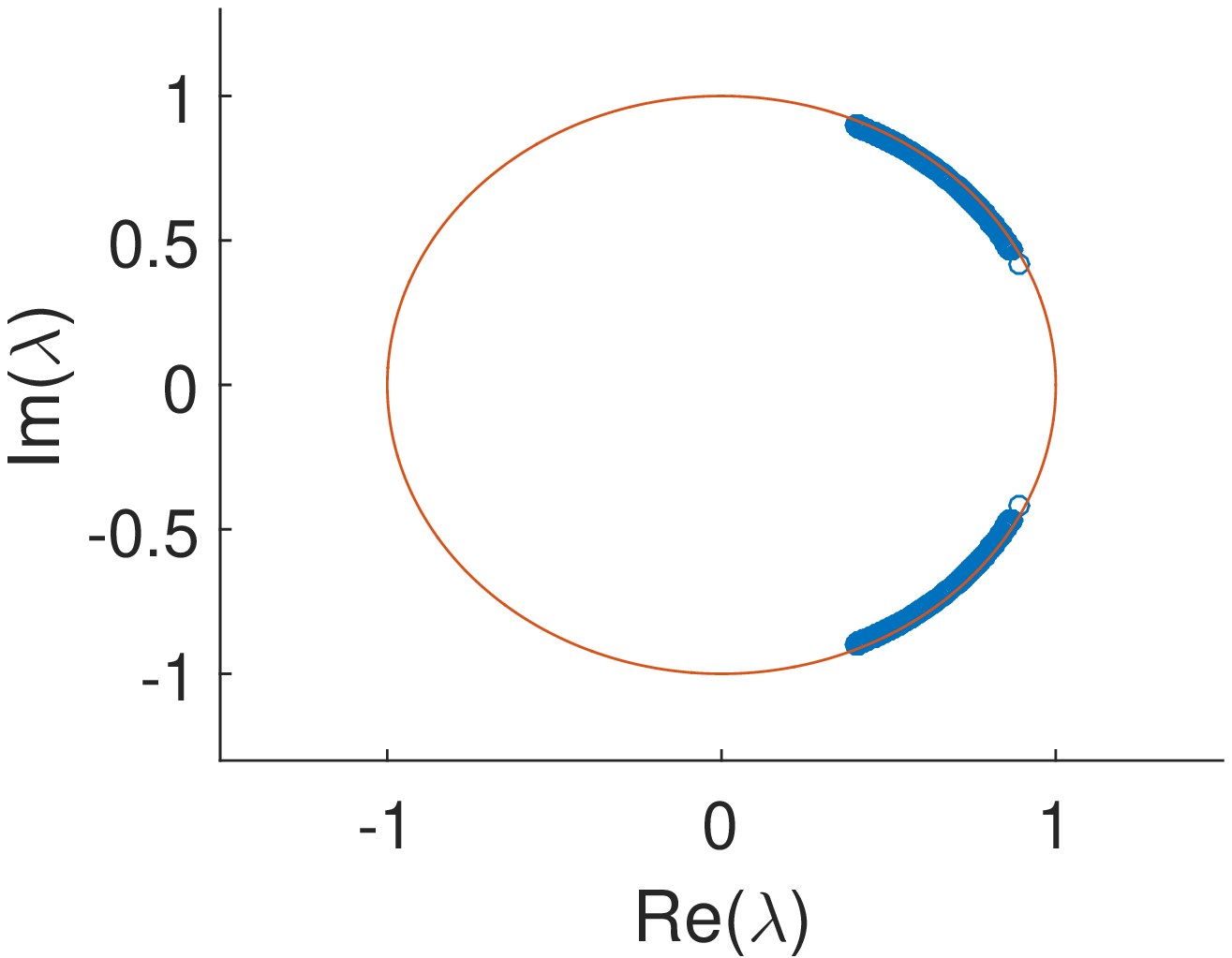}\label{subfig:eig_kg_eps_0_05}}
				\subfigure[]{\includegraphics[scale=0.5]{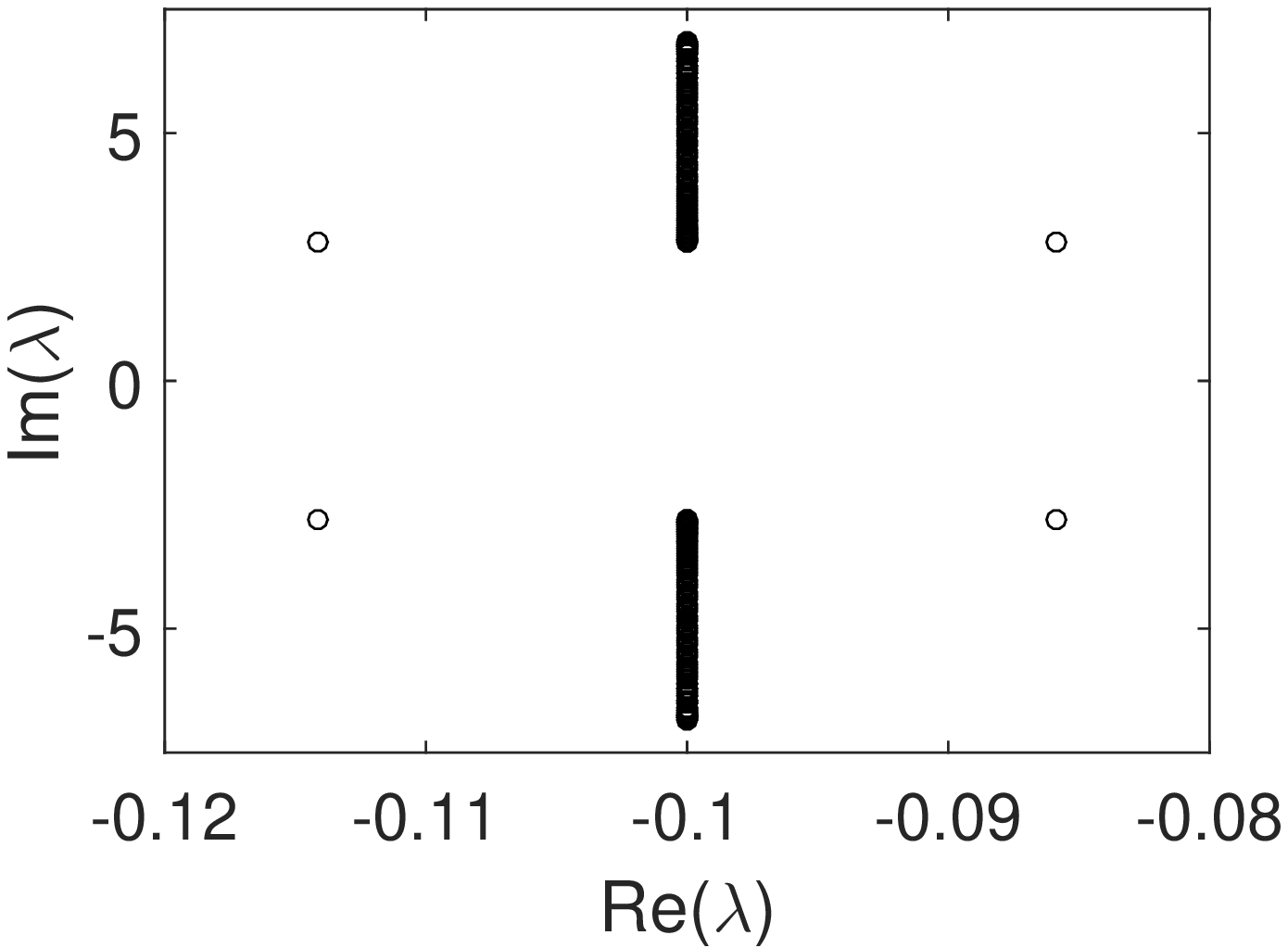}\label{subfig:eig_dnls}}
				\subfigure[]{\includegraphics[scale=0.5]{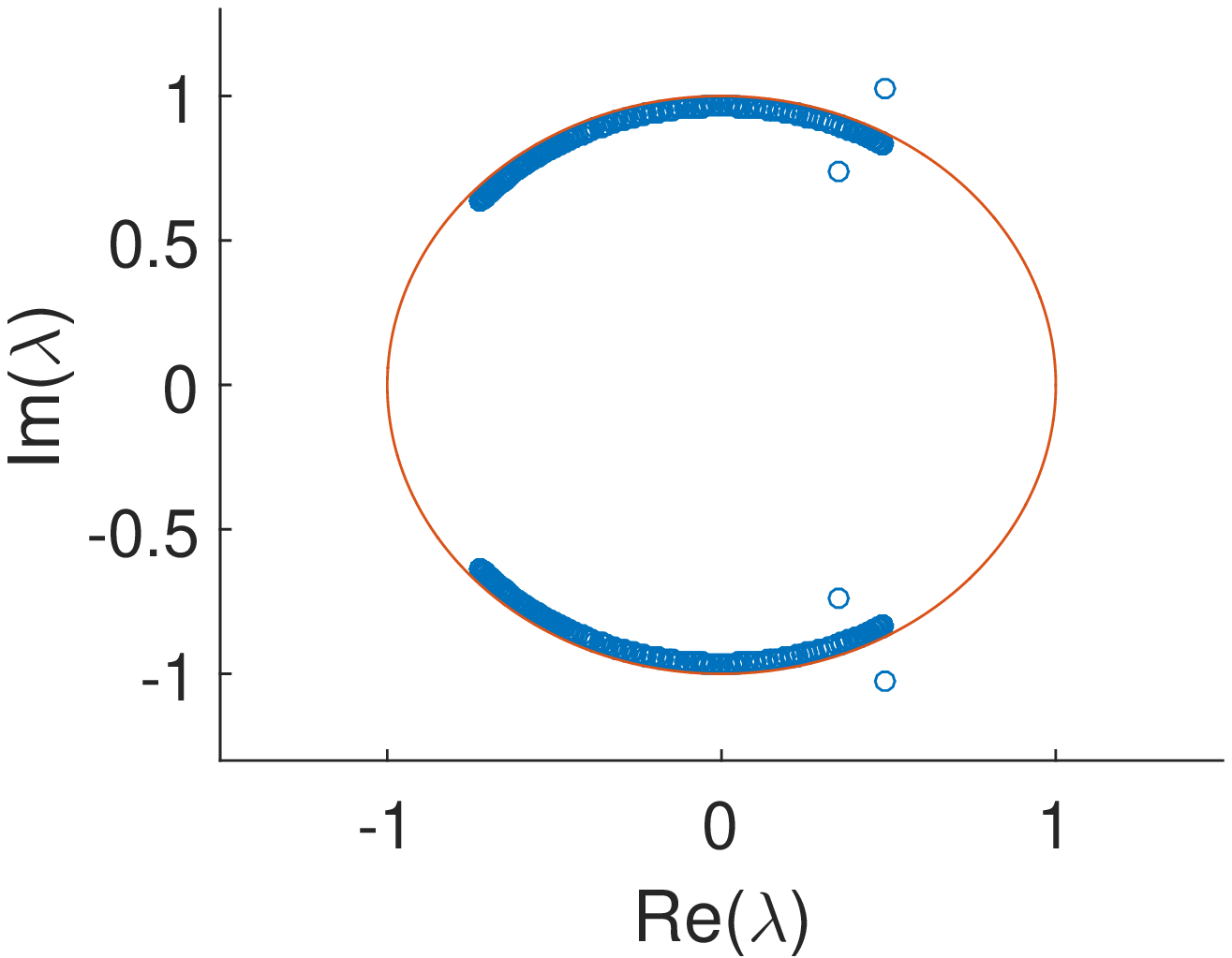}\label{subfig:eig_kg_eps_0_1}}	
				\subfigure[]{\includegraphics[scale=0.5]{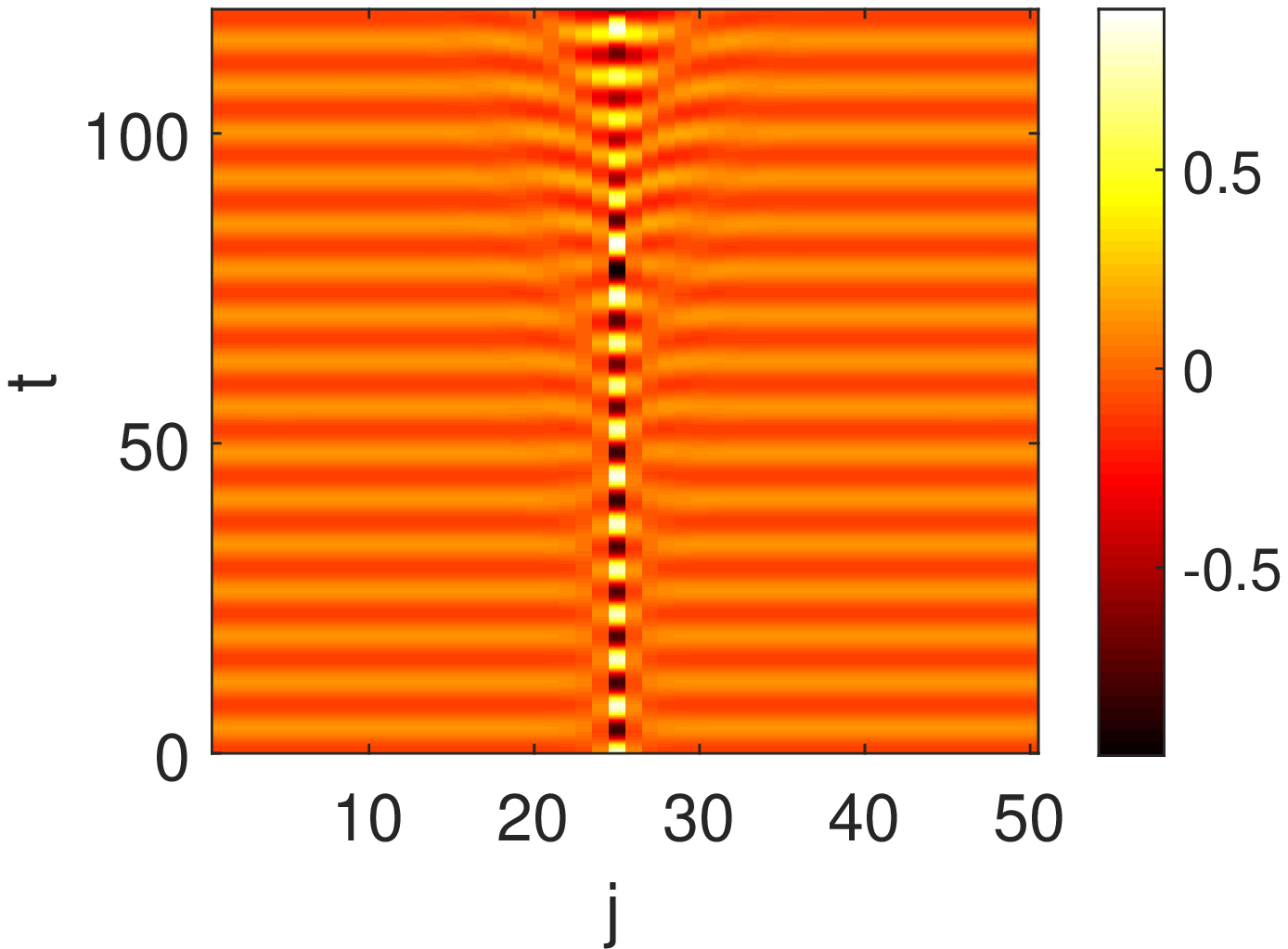}\label{subfig:runge_stability_eps_0_1}}
				\caption{(a) Characteristic multipliers, i.e., eigenvalues of the monodromy matrix, of the breather in Fig.\ \ref{subfig:prof_kg_eps_0_05}, showing the linear stability of the solution. (b) Eigenvalues of the corresponding discrete soliton. Because all of the eigenvalues are on the left half-plane, the solution is linearly stable. (c) The same as panel (a), but for $\epsilon=0.1$, i.e., the breather is linearly unstable. (d) Time dynamics of the unstable breather with multipliers shown in panel (c).
				}
				\label{fig:onsite_stability_test}
			\end{figure}
			
			We present in Fig.\ \ref{subfig:prof_kg_eps_0_05} a breather solution and its time-dynamics in one period for $\epsilon=0.05$. We also compare in Fig.\ \ref{subfig:prof_compare_kg_dnls_eps_0_05} the breather in panel (a) and the approximation \eqref{ansatzerror} where $A_j$ is the discrete soliton solution obtained from solving Eq.\ \eqref{dNLSerror}. One can see the good agreement between them. 
			
			By defining the error between breathers of \eqref{dKGdriven} and the approximation \eqref{ansatzerror} using discrete solitons of \eqref{dNLSerror} as
			\[
			\mathcal{E}=\sup_{0\leq t<2\pi/\Omega}\lVert y(t) \rVert_{\ell^2(\mathbb{Z}_N)},
			\] we plot the error  in Fig.\ \ref{fig:err} for varying $\epsilon$. We also depict in the same picture, the best power fit, which interestingly shows an algebraic power that follows the estimated error in Theorem \ref{theorem1}, i.e., $\sim\epsilon^{3/2}$.
			
			For the sake of completeness, we show in Fig.\ \ref{subfig:eig_kg_eps_0_05} the Floquet multipliers of the solution in Fig.\ \ref{subfig:prof_kg_eps_0_05} for $\epsilon=0.05$, that are obtained from solving the linear equations \eqref{fm}. Because all the eigenvalues are inside the unit circle, the breather is stable. We plot the eigenvalues of the corresponding discrete soliton in Fig.\ \ref{subfig:eig_dnls}
			, also showing stability. Because both solutions are stable, the error between them, that is initially of order $\mathcal{O}(\epsilon^{3/2})$, will stay the same as time evolves until at least $t=2\tau_0/\epsilon,$ for some $\tau_0>0$.
			
			When $\epsilon$ is taken to be larger, we observe that breathers of Eq.\ \eqref{dKGdriven} can become unstable. Shown in Fig.\ \ref{subfig:eig_kg_eps_0_1} are the Floquet multipliers of the breather when $\epsilon=0.1$. Because there is an eigenvalue outside the unit circle, the localised solution is unstable. The unstable eigenvalue bifurcates from the collision of an eigenvalue with the continuous spectrum. Note that the corresponding localised solution of the approximating Schr\"odinger equation \eqref{dNLSerror} is still the same as that shown in Fig.\ \ref{subfig:prof_kg_eps_0_05}, i.e., it is a stable solution. 
			
			We show in Fig.\ \ref{subfig:runge_stability_eps_0_1} the dynamics of the unstable solution. One can observe that it maintains its shape in the form of periodic oscillations for a while, i.e., $t\sim2\tau_0/\epsilon$. After that, the breather starts to deform and break up. Eventually the solution will collapse, i.e.\ unbounded blow-up, which is typical for the Klein-Gordon equation \eqref{dKGdriven}, even when it is undriven \cite{2} (see also a related work \cite{8}). 
			
\section*{Acknowledgements}
				YM thanks MoRA (Ministry of Religious Affairs) Scholarship of the Republic of Indoesia for a financial support. The research of FTA and BEG are supported by  Riset P3MI ITB 2019. RK gratefully acknowledges financial support from Lembaga Pengelolaan Dana Pendidikan (Indonesia Endowment Fund for Education) (Grant No.\ - Ref: S-34/LPDP.3/2017). 

\end{document}